\documentclass[copyright,creativecommons]{eptcs}
\providecommand{\event}{HAS 2013}

\usepackage{amssymb}
\usepackage{amsmath}
\usepackage{graphicx}
\usepackage{amsfonts}
\usepackage{epstopdf}

\usepackage{subfigure}

\newcommand{\Reals}{\mathbb{R}}
   
\newcommand{\Time}{{\sf T}}  

\newcommand{\val}[1]{\mathbf{#1}}
\newcommand{\vals}[1]{\mathit{vals}(#1)}

\newcommand{\ms}[1]{\ifmmode%
\mathord{\mathcode`-="702D\it #1\mathcode`\-="2200}\else%
$\mathord{\mathcode`-="702D\it #1\mathcode`\-="2200}$\fi}

\newcommand{\Tt}{\mathcal{T}}
\newcommand{\M}{\mathcal{M}}
\newcommand{\AutA}{\mathcal{A}}
\newcommand{\AutB}{\mathcal{B}}

\newcommand{\fstate}[1]{#1.\ms{fstate}}
\newcommand{\lstate}[1]{#1.\ms{lstate}}
\newcommand{\arrow}[1]{\mathrel{\stackrel{#1}{\rightarrow}}}

\newcommand{\restr}{\mathop{\lceil}}
\newcommand{\restrrange}{\mathrel{\downarrow}}

\newcommand{\domain}[1]{{\it dom}(#1)}

\newcommand{\concat}{\mathbin{^{\frown}}}
\newcommand{\hexec}[1]{\ms{execs}(#1)}

\newcommand{\htrace}[1]{{\it trace}(#1)}
\newcommand{\htraces}[1]{{\it traces}(#1)}

\newcommand{\deq}{\mathrel{\stackrel{\scriptscriptstyle\Delta}{=}}}

\newcommand{\type}[1]{\ms{type{(#1)}}}
\newcommand{\dtype}[1]{\ms{dtype{(#1)}}}

\newcommand{\as}[1]{#1_a}

\newcommand{\levup}{\uparrow}
\newcommand{\ws}[1]{#1_{w}}

\newcommand{\tran}[1]{\stackrel{#1}{\longrightarrow}}

\newcommand{\padact}{\varepsilon}

\newcommand{\padex}{\gamma}

\newtheorem{theorem}{Theorem}
\newtheorem{lemma}{Lemma}

\newtheorem{definition}{Definition}
\newtheorem{proposition}{Proposition}

\newcommand{\PF}{\par\noindent{\bf Proof:}~}
\newcommand{\QeD}{\hfill{\rule{3mm}{3mm}}\medskip}
\newenvironment{proof}{\PF}{\QeD}

\title{World Automata: a compositional approach to model implicit communication in hierarchical Hybrid Systems. \thanks{This research has been supported by EC-Project C4C ({\em Control for Coordination of Distributed Systems}) and the coordination action EuroSurge (grant no. 288233) funded by the European Commission in the 7th EC framework program. }}
\author{Marta Capiluppi \quad Roberto Segala
\institute{Universit\`a di Verona \\
					Dipartimento di Informatica \\
					Verona, Italy}
\email{marta.capiluppi@univr.it, roberto.segala@univr.it}
}
\def\titlerunning{World Automata}
\def\authorrunning{M. Capiluppi and R. Segala}

\begin{document}

\maketitle

\begin{abstract}
We propose an extension of Hybrid I/O Automata (HIOAs) to model agent systems and their implicit communication through perturbation of the environment, like localization of objects or radio signals diffusion and detection. The new object, called World Automaton (WA), is built in such a way to preserve as much as possible of the compositional properties of HIOAs and its underlying theory. From the formal point of view we enrich classical HIOAs with a set of {\em world variables} whose values are functions both of time and space. World variables are treated similarly to local variables of HIOAs, except in parallel composition, where the perturbations produced by world variables are summed. In such way, we obtain a structure able to model both agents and environments, thus inducing a hierarchy in the model and leading to the introduction of a new operator. Indeed this operator, called inplacement, is needed to represent the possibility of an object (WA) of living inside another object/environment (WA). 
\end{abstract}

%
%

\section{Introduction}

Agents moving in an environment need to communicate to achieve coordination for a common objective. Their communication method can be {\em explicit}, when they broadcast or send a signal to the other agents. This forces the introduction of a mechanism of declaration of intentions, for which either each agent communicates its position at fixed time instants, or a supervisor exists, that is able to know the topology of the network of agents. They can also communicate in an {\em implicit} way, i.e. using their {\em senses} to catch perturbations of the environment due to other agents. Implicit communication may be used in case of presence of noise and environmental hostilities, that prevent direct communication, as well as in case of necessity of not being intercepted or of being subject to faults and failures affecting the sender or the receiver. Using implicit communication, the agents do not need to broadcast their position, because they can {\em feel} the presence of other agents, or objects in general, avoiding collision. Moreover, implicit communication is also used to catch radio (or similar) signals that involve perturbation of the environment via sound waves embedding messages. Implicit communication does not necessarily substitute direct communication, but can be used as a redundant and faster way of communication in case of immediate response to an environment perturbation.

To face the problem of implicit communication, in \cite{gandalf12} we specialized some variables of the Hybrid I/O Automata (HIOAs) of \cite{segala} and called them {\em world variables}. This modification has been motivated by the case studies of the European Project CON4COORD (EU FP7 223844): agents performing a {\em search} mission, such as UAVs \cite{polycarpou} or autonomous underwater vehicles \cite{sousa}, but also of road traffic problems \cite{road,roadcontrol} and autonomous straddle carriers in harbours \cite{adhs12}. Indeed what is common to each case study is the presence of a collection of agents that communicate and coordinate to achieve a common goal. Moreover the agents move within an environment that changes dynamically and detect each other's presence not necessarily via direct communication but rather by observations of the environmental changes. World variables are, hence, used to represent the changes of the environment as perceived by the agents, achieving implicit communication. The difference with the other variables of HIOAs is that world variables dynamics depend both on the time and the space, creating a sort of map of the surrounding world for the agents.

What was missing in \cite{gandalf12} is the interaction between agents and environment. Indeed changes on the environment are caused by agents and the environment has to {\em refer} to the other agents these changes. By adopting the usual compositional rule to represent also the environment with the same model of agents, i.e. automata, it is possible to extend the world variables paradigm. In this paper we will then consider also the environment as a HIOA with world variables. This choice introduces a hierarchy of automata, that can be both and contemporarily agents living in an environment and environments for other agents.

We will call the extended HIOAs {\em World Automata} (WAs), with the double aim of stressing both the capability of representation of this framework (the world itself) and of representing the reality in the most natural way possible, without adding artificial machineries. We renamed the automata because, even if they are an extension of the HIOAs, they will need slightly different operators to prove compositionality results. 
 
The focus of this paper will be on the hierarchical representation of WAs. Indeed, it implies the need to distinguish between the variables used by an automaton to communicate with the world in which it lives and the variables it uses to communicate with the world it creates. 
We could simply partition the sets of variables into variables used to communicate with the outside world and variables used to communicate with the inside world. Nevertheless, this method will hide the hierarchy of automata. On the contrary, we want to keep the hierarchy in order to be able to always retrieve automata at different levels of depth.
For this reason we equip variables with levels and we assume that variables at different levels are distinct. 

The introduction of levels slightly changes the parallel composition policy, since we want to compose automata without losing the original hierarchy of the components. We need to impose that variables at levels not supposed to interact are not synchronized: it suffices to use disjoint sets of world variables, renaming them when needed. Indeed parallel composition requires synchronization only at the same level.
To describe the interaction of two automata at different levels we introduce a new operator called Inplacement, used to compose a WA inside another WA. Inplacement establishes the policy of communication between the environment a WA provides and the WAs living in this environment.

WAs have been used to model a real application in \cite{adhs12}, where straddle carriers autonomously moving into a harbour have to follow some trajectories avoiding collisions. The interested reader can find in the above mentioned paper most of the modeling features described in this paper, including an example of radio communication.

At the best of our knowledge, none of the existing modeling frameworks can be used to formally represent both implicit communication and hierarchy, without flattening the representation. Some approaches have been used to write a language capable to represent dynamically changing systems. One has been introduced in \cite{shift} where dynamic networks of hybrid automata are studied. The introduced programming language focuses on dynamical interfaces. Another method has been presented in \cite{cif} where a compositional interchange format (CIF) defined in terms of an interchange automaton is used as a common language to describe objects from the different models for hybrid systems existing in literature. None of these two languages is based on the idea of implicit communication coded by world variables. Nevertheless there is an ongoing effort to extend the current version of CIF to also include the language generated by world automata, in order to have a tool for automatically implement systems where implicit communication and hierarchy is needed.

The paper is organized as follows: in Section \ref{sec:wa} we introduce the modeling framework; in Section \ref{sec:parallel} parallel composition of WAs is described; in Section \ref{sec:inplace} the inplacement operation is introduced. The theory is illustrated throughout the paper with a very simple example. \\
Notice that all the results presented in the paper use the notation and follow the results of \cite{segala}.

\section{World automata}\label{sec:wa}

In \cite{gandalf12} we extended the HIOA modeling framework of \cite{segala} by specializing some variables, called {\em world variables}.
The main difference between world and standard automaton variables is that the type of world variables is a function of time and space, not only of time as in standard automaton variables. Hence world variables values (and trajectories) will depend both on the instant of time and the position in the underlying space. An automaton $\AutA$ will use its world inputs $\ws{U}$ to receive stimuli from the world it lives in. Analogously it will use its world outputs $\ws{Y}$ to give stimuli to the world it lives in. Finally internal world variables $\ws{X}$ are used to represent the world characteristics of $\AutA$. 
Here we extend the concept of world variables with hierarchy that can be used to represent nested worlds. To preserve the hierarchy and  the identity of each automaton inside a world, we introduce a level function described as $l:S\to\mathbb{N}$ and extracting the level of a variable in any set $S$. Basically, if we consider only variables of level $0$, then we have an ordinary HIOA equipped with some input and output world variables that are used to interact with its external world. The world variables of level 1 describe the world provided by an automaton. The world variables of level 2 describe the world provided by automata of level $1$ and so on.
We will call the HIOA with world variables and levels: World Automaton (WA).


In the following we will assume an underlying topological space $\M$. Without loss of generality the reader can think at $\M$ as a metric space, e.g. $\Reals^3$.
\begin{definition}\label{automaton}
{\em World Automaton}\\
A World Automaton (WA) $\AutA$ is a tuple
\begin{equation*}
((\ws{U},\ws{X},\ws{Y}),(\as{U},\as{X},\as{Y}),(I,H,O),Q,\Theta,D,\mathcal{T},l)
\end{equation*}
where
\begin{itemize}
\item $(\ws{U},\ws{X},\ws{Y})$ are disjoint sets of \emph{world input, inner, and output variables}, respectively. Let $W$ denote the set $\ws{U}\cup\ws{X}\cup\ws{Y}$ of \emph{world variables}.
\item $(\as{U},\as{X},\as{Y})$ are disjoint sets of \emph{automaton input, inner, and output variables}, respectively. Let $U,X,Y$ denote the sets $\ws{U}\cup\as{U},\ws{X}\cup\as{X},\ws{Y}\cup\as{Y}$ of \emph{input, inner, and output variables}, respectively, and let $V$ denote the set $U\cup X\cup Y$ of \emph{variables}.
\item $(I,H,O)$ are disjoint sets of \emph{input, hidden, and output actions}, respectively. Let $A$ denote the set $I\cup H\cup O$ of \emph{actions}.
\item $Q\subseteq \vals{X}$ is the set of \emph{states}.
\item $\Theta\subseteq Q$ is a nonempty set of \emph{initial states}.
\item $D\subseteq \vals{X}\times A \times \vals{X}$ is the \emph{discrete transition relation}.
\item $\Tt$ is a set of trajectories on
$V$ that satisfy the following axioms
\begin{description}
\item[T1] {\em (Prefix closure)} \\
  For every $\tau\in\Tt$ and every $\tau' \leq \tau$, $\tau' \in \Tt$.
\item[T2] {\em (Suffix closure)} \\
  For every $\tau \in \Tt$ and every $t\in\domain{\tau}$, $\tau\unrhd t \in \Tt$.
\item[T3] {\em (Concatenation closure)} \\
  Let $\tau_0 , \tau_1 , \tau_2, \ldots$ be a sequence of trajectories
  in $\Tt$ such that, for each nonfinal index $i$, $\tau_i$ is closed
  and $\lstate{\tau_i} = \fstate{\tau_{i+1}}$.
  Then $\tau_0 \concat \tau_1 \concat \tau_2 \cdots \in \Tt$.
\end{description}
\item $l:S\to\mathbb{N}$ is the {\em level function} extracting the level of a variable or action in any set $S$.
\end{itemize}
\end{definition}



In this description, like in the HIOA theory, variables $(\as{U},\as{X},\as{Y})$ are used by the automaton to communicate with other automata living in the same world. To keep the theory consistent with previous descriptions of automata, all the $X$ variables represent persistent characteristics of the system. 
For the sake of simplicity we will not use the level function in the rest of the paper, but we will make the variable level explicit by writing $v[i]$ for the variable of level $i$ and $S[i]$ for the variables in $S$ whose level is $i$. 
Moreover we will consider only finite depth WAs in this paper. 
Note that $\ws{X}[0] = \emptyset$, due to the definition of $X$ variables above and to the fact that level 0 is meant to be the level of the outside world, i.e. the level in which the automaton lives. 

{\bf Notation}: For each variable $v$, we
assume both a {\em (static) type\/}, \type{v}, which gives the set of values it
may take on, and a {\em dynamic type}, \dtype{v}, which gives the
set of trajectories it may follow. 
A {\em valuation\/} $\val{v}$ for a set of variables $V$ is a function
that associates with each variable $v \in V$ a value in $\type{v}$.
Let $J$ be a left-closed interval of $\Time$ (the time axis) with left endpoint equal to
$0$. Then a {\em $J$-trajectory\/} for $V$ is a function $\tau:
J\rightarrow\vals{V}$,
such that for each $v \in V$, $\tau\restrrange v \in\dtype{v}$.
A {\em trajectory\/} for $V$ is a $J$-trajectory for $V$, for any
$J$. Trajectory $\tau$ is a {\em prefix\/} of trajectory $\tau'$, denoted
by $\tau \leq \tau'$, if $\tau$ can be obtained by restricting
$\tau'$ to a subset of its domain. We define $\tau \unrhd t  \deq  (\tau \restr [t,\infty)) - t$. The concatenation $ \concat$ of two trajectories is obtained
by taking the union of the first trajectory and the function obtained
by shifting the domain of the second trajectory until the start time
agrees with the limit time of the first trajectory; the last valuation
of the first trajectory, which may not be the same as the first valuation
of the second trajectory, is the one that appears in the concatenation.
Prefix, suffix and concatenation operations return trajectories. 
We write $f \restr P$ for the restriction of function $f$ to set $P$,
that is, the function $g$ with $\domain{g} = \domain{f} \cap P$ such that
$g(c) = f(c)$ for each $c \in \domain{g}$. 
If $f$ is a function whose range is a set of functions and $P$ is a set,
then we write $f \restrrange P$ for the function $g$ with $\domain{g}
= \domain{f}$ such that $g(c) = f(c) \restr P$ for each $c \in \domain{g}$. 
For more detail the interested reader can refer to \cite{segala}.

 For a set of objects $S$ we will write $S[i,i+1]$ to indicate the objects of $S$ at level $i$ and $i+1$. We will also write $S[i,.]$ for objects of $S$ at levels $i$ and greater (i.e. deeper).\\
With some minor extensions \cite{tec_rep}, the results on semantics of HIOAs are still valid for WAs, because they have been designed to follow as much as possible the HIOA theory. Hence all the results on executions, traces and simulation presented in \cite{segala} are extended to WAs.

\section{Parallel Composition} \label{sec:parallel}

In our framework, Parallel Composition models the interaction and communication of two or more agents living in the same world, i.e. of two WAs at the same level, with the environment, i.e. the world outside. We extend by comparison the notion of parallel composition introduced in \cite{gandalf12} for HIOAs with world variables, with the treatment of levels.
First, we introduce compatibility conditions to prevent undesired interactions between different levels for the WAs that have to be composed. 

\begin{definition}\label{def:parcomp}
Two WAs $\AutA_1$ and $\AutA_2$ are compatible if 
\begin{enumerate}
\item $V_1[1,.]\cap V_2[1,.] = \emptyset$, $A_1[1,.]\cap A_2[1,.] = \emptyset$,
\item $(U_{w1}\cup U_{w2}) \cap (Y_{w1}\cup Y_{w2}) = \emptyset$.
\item $H_1\cap A_2 = H_2\cap A_1 = \emptyset$,
\item $X_1\cap V_2 = X_2\cap V_1 = \emptyset$,
\item $O_1\cap O_2 = \emptyset$,
\item $Y_1\cap Y_2 = \emptyset$.
\end{enumerate}
\end{definition}
The reader may notice that conditions 2 to 6 are the compatibility conditions for HIOAs with world variables of \cite{gandalf12}. 
The only difference is given by the first condition that states that all inner levels (higher than 0) are disjoint, and therefore no communication can occur at such levels. This means that generated worlds are disjoint. 
Since by the first condition communication may occur only at level 0, all the other properties are interesting only for level 0, even if not specified.
These conditions, when not satisfied by the WAs, can be obtained by changing variables names.
\begin{definition}\label{def:parallel}
{\em Parallel composition}\\
If $\AutA_1,\AutA_2$ are two compatible WAs, then their composition $\AutA_1\|\AutA_2$ is defined as the structure\\ $\AutA=(\ws{U},\ws{X},\ws{Y},\as{U},\as{X},\as{Y},I,H,O,Q,\Theta,D,\Tt,l)$ with:
\begin{enumerate}
\item $\ws{U} = U_{w1}\cup U_{w2}$, $\ws{X} = X_{w1}\cup X_{w2}$, $\ws{Y} = Y_{w1}\cup Y_{w2}$
\item $\as{Y}=Y_{a1}\cup Y_{a2}$, $\as{X}=X_{a1}\cup X_{a2}$, $\as{U}=(U_{a1}\cup U_{a2})\setminus\as{Y}$
\item $O=O_1\cup O_2$, $I=(I_1\cup I_2)\setminus O$ and $H=H_1\cup H_2$
\item $Q=\{\val{x}\in \vals{X}\mid\val{x}\restr X_1\in Q_1\land \val{x}\restr X_2\in Q_2\}$
\item $\Theta=\{\val{x}\in Q \mid \val{x}\restr X_1\in \Theta_1\land \val{x}\restr X_2\in \Theta_2\}$
\item $D = \{(\val{x},a,\val{x}') \mid $ for each $i\in\{1,2\}$\\
\hspace*{7em} either $a\in A_i$ and $\val{x}\restr X_i \tran{a} \val{x}^{\prime}\restr X_i $,\\
\hspace*{7em} or $a\notin A_i$ and $\val{x}\restr X_i = \val{x}^{\prime}\restr X_i\}$.
\item $\Tt=\{\tau\mid$ there exists $\tau_1\in\Tt_1,\tau_2\in\Tt_2$ such that\\ 
$\tau\restrrange (V_i\setminus (Y_{w1}\cap Y_{w2}))=\tau_i\restrrange (V_i\setminus  (Y_{w1}\cap Y_{w2})), i\in\{1,2\}$\\
$\tau\restrrange  (Y_{w1}\cap Y_{w2})=\tau_1\restrrange  (Y_{w1}\cap Y_{w2})+\tau_2\restrrange (Y_{w1}\cap Y_{w2})\}$
\item $l(v)=l_1(v)$ if $v\in V_1$, $l(v)=l_2(v)$ if $v\in V_2$.
\end{enumerate}
\end{definition}
This definition of parallel composition differs from the one of HIOAs with world variables of \cite{gandalf12} only by the last condition:
it preserves levels of the variables in composed automata $\AutA_1,\AutA_2$.
The following result on composability is proved:
\begin{proposition}\label{th:autparallel}
The composition of two WAs is a WA.
\end{proposition}

Proof of Theorem \ref{th:autparallel} and all the results on parallel composition reported with their proofs in \cite{gandalf12} are still valid for WAs, because the introduction of levels does not affect parallel composition, due to the compatibility conditions.

\section{Inplacement operator} \label{sec:inplace}

The main difference with the HIOA theory and the framework presented in \cite{gandalf12} is that WAs create a hierarchy of automata. To this end a second operator is introduced that represents the interaction of a WA $\AutA_2$ inside another WA $\AutA_1$ with the world created by $\AutA_1$. The result is equivalent to compose $\AutA_2$ with the automata of level 1 of $\AutA_1$, although there are some important differences. This operation shows how the hierarchical communication works between the automata inside a world and the automaton representing their external world. Note that with this operator we do not want to describe the action of a WA {\em moving} inside another WA, but we want to describe the static behavior of an automaton {\em inside} another.\\
\noindent {\bf Notation}: For a set $S$ of objects we write $S\levup$ for the set obtained from $S$ by increasing by $1$ the level of each object. 
\begin{definition}\label{def:incomp}
An automaton $\AutA_2$ is inplace compatible with $\AutA_1$ if, letting $\AutA_3$ be $\AutA_2\levup$,
\begin{enumerate}
\item $V_1[2,.]\cap V_3[2,.] = \emptyset$, $A_1[2,.]\cap A_3[2,.] = \emptyset$,
\item $Y_{w1}\cap Y_{w3} = \emptyset$,
\item $H_1\cap A_3 = H_3\cap A_1 = \emptyset$,
\item $X_1\cap V_3 = X_3\cap V_1 = \emptyset$.
\item $O_1\cap O_3 = \emptyset$,
\item $Y_1\cap Y_3 = \emptyset$.
\end{enumerate}
\end{definition}
As for parallel composition, when using the {\em inplacement} operator the two composing automata need to have the disjoint sets of variables at levels deeper than 1. No compatibility condition is stated for level 0 because only the outside automaton $\AutA_1$ has this level: indeed the level of $\AutA_2$ is increased to be composed with automata of level 1 of $\AutA_1$. Disjointness of variables is again achieved by renaming when necessary. 
Since the world variables of level 1 are used to let $\AutA_2$ communicate with the world provided by $\AutA_1$, there should be no conflicts on outputs (condition 2).
\begin{definition}\label{def:inplace}
The inplacement of WA $\AutA_2$ inside $\AutA_1$, denoted by $\AutA_1[\AutA_2]$, is a system\\ $\AutA=(\ws{U},\ws{X},\ws{Y},\as{U},\as{X},\as{Y},I,H,O,Q,\Theta,D,\Tt,l)$ where, letting $\AutA_3$ be $\AutA_2\levup$,
\begin{enumerate}
\item $\ws{U}[0,1] = U_{w1}[0,1]$, $\ws{X}[0,1] = X_{w1}[0,1]$, $\ws{Y}[0,1] = Y_{w1}[0,1]$
\item $\ws{U}[2,.] = (U_{w1}\cup U_{w3})[2,.]$, $\ws{X}[2,.] = (X_{w1}\cup X_{w3})[2,.]$, $\ws{Y}[2,.] = (Y_{w1}\cup Y_{w3})[2,.]$
\item $\as{Y}=Y_{a1}\cup Y_{a3}$, $\as{U}=(U_{a1}\cup U_{a3})\setminus\as{Y}$ and $\as{X}=X_{a1}\cup X_{a3}$
\item $O=O_1\cup O_3$, $I=(I_1\cup I_3)\setminus O$ and $H=H_1\cup H_3$
\item $Q=\{\val{x}\in \vals{X}\mid\val{x}\restr X_1\in Q_1\land \val{x}\restr X_3\in Q_3\}$
\item $\Theta=\{\val{x}\in Q \mid \val{x}\restr X_1\in \Theta_1\land \val{x}\restr X_3\in \Theta_3\}$
\item $D = \{(\val{x},a,\val{x}') \mid $ for each $i\in\{1,3\}$\\
either $a\in A_i$ and $\val{x}\restr X_i \tran{a} \val{x}^{\prime}\restr X_i $,\\
or $a\notin A_i$ and $\val{x}\restr X_i = \val{x}^{\prime}\restr X_i\}$.
\item $\Tt=\{\tau\mid$ there exists $\tau_1\in\Tt_1,\tau_3\in\Tt_3$ such that\\
$\tau\restrrange (V_i\setminus (U_{w1}\cap Y_{w3}))=\tau_i\restrrange (V_i\setminus (U_{w1}\cap Y_{w3}), i\in\{1,3\}$.\\
$\tau_1\restrrange  (U_{w1}\cap Y_{w3})=\tau\restrrange  (U_{w1}\cap Y_{w3})+\tau_3\restrrange (U_{w1}\cap Y_{w3})\}.$\\
For each $u\in U_{w3}\setminus V$, each $t\in\domain{\tau_3}$, $\tau_3(u)(t) = 0$.
\item $l(v)=l_1(v)$ if $v\in V_1$, $l(v)=l_3(v)$ if $v\in V_3$.
\end{enumerate}
\end{definition}
Note that the set of world variables of levels 0 and 1 are taken only from $\AutA_1$. For level 0  this definition derives from the fact that $\AutA_3$ has no objects at level 0; for level $1$ it derives from the fact that $\AutA_3$ has no internal world variables at level 1 and that the world variables of level $1$ of $\AutA_3$ that are not captured by $\AutA_1$ are no longer necessary since the environment does not use them.  This is also expressed by the last condition in the definition of the set of trajectories, which makes sure that the world input of $\AutA_3$ that is not captured by $\AutA_1$ is always $0$. This last condition is used to avoid the risk that, composing $\AutA_1$ with other WAs, the input of its inside automata not captured by $\AutA_1$ are then captured by the other composing WAs. Indeed, local worlds of automata have to be kept separated in parallel composition.
Note also that, since from level 2 on, $V_1\cap V_3=\emptyset$, and $\AutA_3$ has no level 0, we have that $U_{w1}\cap Y_{w3}$ might be $\neq\emptyset$ only at level 1.\\
The main difference with parallel composition is given by the condition on trajectories. If we consider the domain of trajectories as a group, we have that this last condition can be expressed as $\tau\restrrange  (U_{w1}\cap Y_{w3})=\tau_1\restrrange  (U_{w1}\cap Y_{w3})-\tau_3\restrrange (U_{w1}\cap Y_{w3})$. This statement derives from the consideration that $\AutA_1$ is supposed to have some input world variables whose values take into account the possibility to have other automata inside. The object resulting from the inplacement composition has, for the same variables, values resulting from the difference between the input signal of the outside automaton and the output signals of local automata. With the reverse reasoning $\AutA_1$ will have some input world variables whose values are the result of the sum of values of the input world variables of the object $\AutA_1[\AutA_2]$ and the values of the output world variables of $\AutA_2$ communicating with them. In an abstract way we can think of $\AutA_1[\AutA_2]$ as an object already containing other WAs, so that the local input world variables are updated decreasing their values each time another WA is composed with the already present WAs of the local world. 

We report here some lemmas on trajectories that will be used in the following proofs.
\begin{lemma}\label{lm:restrprojtraj}
Let $\tau$ be a trajectory in $V$. Let $I\subseteq\domain{\tau}$ and $V^\prime\subseteq V$. Then $(\tau\restr I)\restrrange V'=(\tau\restrrange V')\restr I$.
\end{lemma}
\begin{lemma}\label{lm:restrsuffixtraj}
Let $\tau$ be a trajectory in $V$. Let $V^\prime\subseteq V$. Then $(\tau \unrhd t)\restrrange V'=(\tau\restrrange V') \unrhd t$.
\end{lemma}
\begin{lemma}\label{lm:restrconcattraj}
Let $\tau$ be a trajectory in $V$ such that $\tau=\tau_0\concat\tau_1\concat\tau_2\concat\ldots$. Let $V^\prime\subseteq V$. Then $(\tau_0\concat\tau_1\concat\tau_2\concat\ldots)\restrrange V'=(\tau_0\restrrange V')\concat(\tau_1\restrrange V')\concat(\tau_2\restrrange V')\concat\ldots$.
\end{lemma}
The proofs of these results are reported in Section 2 of \cite{tec_rep}.

Without loss of generality in the following we will assume that the domain of trajectories is a group. The results we are now going to introduce can be proved also using a monoid. We use the group to be coherent with the results of parallel composition. 
\begin{proposition}\label{th:autinside}
The inplacement of WA $\AutA_2$ inside WA $\AutA_1$ is a WA.
\end{proposition}
\begin{proof}
We show that $\AutA_1[\AutA_2]$ satisfies the properties of a WA. Again we let $\AutA_3$ be $\AutA_2\levup$.
Disjointness of the $U,X,Y$ sets follows from disjointness of the same sets in $\AutA_1$ and $\AutA_3$ and compatibility. Similarly for the actions. Nonemptiness of starting state follows from nonemptiness of starting states of $\AutA_1$ and $\AutA_3$ and disjointness of $X_1$ and $X_3$.
We verify the {\bf T} properties of trajectories. Let $C_{13}$ be $U_{w1}\cap Y_{w3}$.
 \begin{description}
  \item[T1]  We want to prove that for every $\tau\in\Tt$ and every $\tau' \leq \tau$, $\tau' \in \Tt$. Let $\tau$ be a trajectory in $\Tt$. Let $i\in\{1,3\}$. By the definition of inplacement there exists $\tau_1\in \Tt_1,\tau_3\in\Tt_3$ such that
  $\tau\restrrange (V_i\setminus C_{13})=\tau_i\restrrange (V_i\setminus C_{13}),$
$\tau\restrrange  C_{13}=\tau_1\restrrange C_{13}-\tau_3\restrrange C_{13}.$
Let $\tau^\prime\leq\tau$. By definition of prefix we have that $\tau^\prime=\tau\restr I$ with $I=\domain{\tau^\prime}\subseteq\domain{\tau}$. Hence we can state that 
$\tau^\prime\restrrange(V_i\setminus C_{13})=(\tau\restr I)\restrrange (V_i\setminus C_{13}).$
By lemma \ref{lm:restrprojtraj}
$(\tau\restr I)\restrrange (V_i\setminus C_{13})=(\tau\restrrange (V_i\setminus C_{13}))\restr I.$
By definition of parallel composition and again by lemma \ref{lm:restrprojtraj}
$(\tau\restrrange (V_i\setminus C_{13}))\restr I=(\tau_i\restrrange (V_i\setminus C_{13}))\restr I=(\tau_i\restr I)\restrrange (V_i\setminus C_{13}).$
Let $\tau^\prime_1=\tau_1\restr I$ and $\tau^\prime_3=\tau_3\restr I$, then
$(\tau_i\restr I)\restrrange (V_i\setminus C_{13})=\tau^\prime_i\restrrange (V_i\setminus C_{13}).$
Analogously, for the second statement of parallel composition of trajectories we have that
$\tau^\prime\restrrange C_{13}=(\tau\restr I)\restrrange C_{13}=(\tau\restrrange C_{13})\restr I=(\tau_1\restrrange C_{13})\restr I-(\tau_3\restrrange C_{13})\restr I=(\tau_1\restr I)\restrrange C_{13}-(\tau_3\restr I)\restrrange C_{13}=\tau^\prime_1\restrrange C_{13}-\tau^\prime_3\restrrange C_{13}.$ Hence $\tau'\in\Tt$.

  \item[T2] We want to prove that for every $\tau \in \Tt$ and every $t\in\domain{\tau}$, $\tau\unrhd t \in \Tt$. Let $\tau$ be a trajectory in $\Tt$. Let $i\in\{1,3\}$. By the definition of inplacement there exists $\tau_1\in \Tt_1,\tau_3\in\Tt_3$ such that
  $\tau\restrrange (V_i\setminus C_{13})=\tau_i\restrrange (V_i\setminus C_{13}), $
$\tau\restrrange  C_{13}=\tau_1\restrrange C_{13}-\tau_3\restrrange C_{13}.$
  Hence, since $\domain{\tau_1}=\domain{\tau_3}=\domain{\tau}$ and by lemma \ref{lm:restrsuffixtraj} we have that 
 $(\tau\unrhd t)\restrrange (V_i\setminus C_{13})=(\tau\restrrange (V_i\setminus C_{13}))\unrhd t=(\tau_i\restrrange (V_i\setminus C_{13}))\unrhd t= (\tau_i\unrhd t)\restrrange (V_i\setminus C_{13})$.
  Moreover 
$(\tau\unrhd t)\restrrange C_{13}=(\tau\restrrange C_{13})\unrhd t
 =(\tau_1\restrrange C_{13})\unrhd t-(\tau_3\restrrange C_{13})\unrhd t=
 (\tau_1\unrhd t)\restrrange C_{13}-(\tau_3\unrhd t)\restrrange C_{13}$.
  By axiom {\bf T2} applied to $\AutA_1$ and $\AutA_3$, for each $t\in\domain{\tau_1}, \tau_1\unrhd t\in\Tt_1$ and for each $t\in\domain{\tau_3}, \tau_3\unrhd t\in\Tt_3$. Hence $\tau\unrhd t \in \Tt$.

  \item[T3] We want to prove that set $\Tt$ is closed under concatenation. Let $\tau_0,\tau_1,\tau_2,\ldots$ be a sequence of trajectories in $\Tt$, such that, for each nonfinal index $j$ $\tau_j$ is closed and  $\lstate{\tau_j}=\fstate{\tau_{j+1}}$. Let $\tau$ be $\tau_0\concat\tau_1\concat\tau_2\concat\ldots$. 
  Let $i\in\{1,3\}$. By definition of inside operator for each $\tau_j$, $\exists \tau_{1j}, \tau_{3j}$ such that
  $\tau_j\restrrange (V_i\setminus C_{13})=\tau_{ij}\restrrange (V_i\setminus C_{13}),$ 
$\tau_j\restrrange  C_{13}=\tau_{1j}\restrrange C_{13}-\tau_{3j}\restrrange C_{13}.$
Let $\tau_i$ be $\tau_{i0}\concat\tau_{i1}\concat\tau_{i2}\concat\ldots, i\in\{1,3\}$. Hence by lemma \ref{lm:restrconcattraj}
$\tau\restrrange(V_i\setminus C_{13})=(\tau_0\restrrange (V_i\setminus C_{13}))\concat(\tau_1\restrrange (V_i\setminus C_{13}))\concat
(\tau_2\restrrange (V_i\setminus C_{13}))\concat\ldots=
(\tau_{i0}\restrrange (V_i\setminus C_{13}))\concat(\tau_{i1}\restrrange (V_i\setminus C_{13}))\concat
(\tau_{i2}\restrrange (V_i\setminus C_{13}))\concat\ldots=
(\tau_{i0}\concat\tau_{i1}\concat\tau_{i2}\concat\ldots)\restrrange (V_i\setminus C_{13})=\tau_i\restrrange (V_i\setminus C_{13})$.
Moreover
$\tau\restrrange  C_{13}=(\tau_0\restrrange  C_{13})\concat(\tau_1\restrrange  C_{13})\concat(\tau_2\restrrange  C_{13})\concat\ldots=
(\tau_{10}\restrrange C_{13}-\tau_{30}\restrrange C_{13})\concat(\tau_{11}\restrrange C_{13}-\tau_{31}\restrrange C_{13})\concat
(\tau_{12}\restrrange C_{13}-\tau_{32}\restrrange C_{13})\concat\ldots=
((\tau_{10}\restrrange C_{13})\concat(\tau_{11}\restrrange C_{13})\concat(\tau_{12}\restrrange C_{13})\concat\ldots)-
((\tau_{30}\restrrange C_{13})\concat(\tau_{31}\restrrange C_{13})\concat(\tau_{32}\restrrange C_{13})\concat\ldots)=
(\tau_{10}\concat\tau_{11}\concat\tau_{12}\concat\ldots)\restrrange C_{13} - (\tau_{30}\concat\tau_{31}\concat\tau_{32}\concat\ldots)\restrrange C_{13}
=\tau_1\restrrange C_{13}-\tau_3\restrrange C_{13}$.
Hence $\tau \in \Tt$.
 
 \end{description}
\end{proof}

{\bf Notation}: Executions of WAs are defined as executions of HIOAs: an \emph{execution fragment} of a WA $\AutA$ is an ($A,V$)-sequence $\alpha=\tau_0 a_1 \tau_1 a_2 \tau_2 \ldots$, where $a_i\in A$, $\tau_i\in \Tt$; if $\tau_i$ is not the last trajectory of $\alpha$, then $\lstate{\tau_i} \arrow{a_{i+1}} \fstate{\tau_{i+1}}$. An execution fragment $\alpha$ is defined to be an {\em execution\/} if $\fstate{\alpha}$ is a start state, that is, $\fstate{\alpha} \in \Theta$. Results on executions of HIOAs are valid also for WAs. 
A {\em trace} of an execution fragment $\alpha$ captures the external behavior of $\AutA$, i.e. what it is needed to identify an automaton from outside. Calling $E=I\cup O$, $Z=U\cup Y$, a trace of a WA $\AutA$ is, then, the ($E,Z$)-restriction of $\alpha$. We will call trace the ($E,Z$)-restriction of $\alpha$ at all levels, supposing that the external behavior is captured at all levels of the automaton. When needed, it is possible to restrict the behavior to a specific level $i$ by defining the ($E[i],Z[i]$)-restriction of $\alpha$, called $[i]$-trace. 
\begin{definition}\label{def:comp_imp}
Automata $\AutA_1$ and $\AutA_2$ are {\em comparable\/} if they have
the same external interface, that is, if world and local input and output sets of variables of $\AutA_1$ are equal to the corresponding sets of $\AutA_2$ and $E_1 = E_2$ at all levels.
If $\AutA_1$ and $\AutA_2$ are comparable then we say that
$\AutA_1$ {\em implements} $\AutA_2$, denoted by $\AutA_1 \leq \AutA_2$, if
$\htraces{\AutA_1} \subseteq \htraces{\AutA_2}$.
\end{definition}

In the following we state the results analogous to lemma 9 and proposition 2 of \cite{gandalf12} for inplacement. Note that the proofs of these results are analogous to the ones for the corresponding results in parallel composition (see also \cite{segala} and \cite{tec_rep}).
\begin{lemma}\label{lm:inexec}
Let $\AutA=\AutA_1[\AutA_2]$, let $\alpha$ be an execution fragment of $\AutA$ and let $\AutA_3$ be $\AutA_2\levup$. Then $\exists \alpha_1,\alpha_3$ execution fragments of $\AutA_1$ and $\AutA_3$ respectively, such that 
\begin{enumerate}
\item $\alpha\restr(A_i,V_i\setminus C_{13})=\alpha_i\restr(A_i,V_i\setminus C_{13}), i=1,3$, and 
\item $\alpha\restr(\emptyset, C_{13})=\alpha_1\restr(\emptyset, C_{13})-\alpha_3\restr(\emptyset, C_{13})$,
\end{enumerate} 
with $C_{13}=U_{w1}\cap Y_{w3}$.
\end{lemma}
\begin{proposition}
  \label{lemma-htraces-proj-in}
Let $\AutA=\AutA_1[\AutA_2]$, let $\beta$ be a trace of $\AutA$ and let $\AutA_3$ be $\AutA_2\levup$. Then $\exists \beta_1,\beta_3$ traces of $\AutA_1,\AutA_3$ respectively, such that 
\begin{enumerate}
\item $\beta\restr(E_i,Z_i\setminus C_{13})=\beta_i\restr(E_i,Z_i\setminus C_{13}), i=1,3$ and 
\item $\beta\restr(\emptyset, C_{13})=\beta_1\restr(\emptyset, C_{13})-\beta_3\restr(\emptyset, C_{13})$,
\end{enumerate} 
with $C_{13}=U_{w1}\cap Y_{w3}$.
\end{proposition}

The following theorems state the substitutivity properties of inplacement for implementation. 
\begin{theorem}
\label{th:sub_imp_in1}
Let $\AutA_1$ and $\AutA_2$ be comparable WAs with $\AutA_1 \leq \AutA_2$.
Let $\AutB$ be a WA compatible with each of $\AutA_1$ and $\AutA_2$. Then $\AutA_1 [\AutB]$ and $\AutA_2 [\AutB]$ are comparable and $\AutA_1 [\AutB] \leq \AutA_2 [\AutB]$.
\end{theorem}
\begin{proof}
Let $\alpha$ be an execution of $\AutA_1[\AutB]$. By lemma \ref{lm:inexec}, two executions $\alpha_1, \alpha_B$ exist, such that $\alpha_1\in\hexec{\AutA_1}$, $\alpha_B\in\hexec{\AutB\levup}$ and:
$\alpha\restr(A_1,V_1\setminus C_{1B})=\alpha_1\restr(A_1,V_1\setminus C_{1B})$,
$\alpha\restr(A_B,V_B\setminus C_{1B})=\alpha_B\restr(A_B,V_B\setminus C_{1B})$, 
$\alpha\restr(\emptyset, C_{1B})=\alpha_1\restr(\emptyset, C_{1B})-\alpha_B\restr(\emptyset, C_{1B})$,
with $C_{1B}=U_{w1}\cap Y_{wB}$.
By lemma \ref{lm:padbuild} we can take paddings of $\alpha, \alpha_1, \alpha_B$ such that the $j^{th}$ trajectory has the same length for all $j$. Let these paddings be $\padex, \padex_1, \padex_B$ respectively with $\padex=\tau_0 a_1 \tau_1 a_2 \tau_2 a_3 \ldots$, $\padex_1=\tau_{01} a'_{1} \tau_{11} a'_{2} \tau_{21} a'_{3} \ldots$ and $\padex_B=\tau_{0B} a''_{1} \tau_{1B} a''_{2} \tau_{2B} a''_{3} \ldots$.
Since $\AutA_1 \leq \AutA_2$ and by compatibility, we can find an execution $\alpha_2$ of $\AutA_2$ with the same trace of $\alpha_1$ and a padding of $\alpha_2$ following lemma \ref{lm:padbuild}. We write $\padex_2=\tau_{02} a'''_{1} \tau_{12} a'''_{2} \tau_{22} a'''_{3} \ldots$. 
By the definition of composition the execution of $\AutA_2[\AutB]$ obtained by $\padex_2$ and $\padex_B$ will be $\padex'=\tau'_0 b_1 \tau'_1 b_2 \tau'_2 b_3 \ldots$, where
$\tau'_j\restrrange (V_2\setminus C_{2B})=\tau_{j2}\restrrange (V_2\setminus  C_{2B})$,
$\tau'_j\restrrange (V_B\setminus C_{2B})=\tau_{jB}\restrrange (V_B\setminus  C_{2B})$,
$\tau'_j\restrrange  C_{2B}=\tau_{j2}\restrrange  C_{2B}-\tau_{jB}\restrrange C_{2B}\}$,
where $C_{2B}=U_{w2}\cap Y_{wB}$. This is valid even if the trajectories in the padded executions have not the length of the original trajectories, by definition of prefix of a trajectory and prefix closure of trajectories in a WA.
Actions $b_i$ might be different, but by construction, compatibility and lemma \ref{lm:padtrace} we have that $\padex'$ has the same trace of $\padex$ hence of $\alpha$. Indeed the (padded) executions can differ only in their internal variables (state), but they do not influences the traces (external variables). For this reason we can state that $\htraces{\AutA_1[\AutB]}\subseteq\htraces{\AutA_2[\AutB]}$, hence, by definition of implementation, $\AutA_1 [\AutB] \leq \AutA_2 [\AutB]$.
\end{proof}

\begin{lemma}\label{lm:levupimp}
Let $\AutB_1$ and $\AutB_2$ be comparable WAs with $\AutB_1 \leq \AutB_2$. Then $(\AutB_1\levup) \leq (\AutB_2\levup)$.
\end{lemma}
\begin{proof}
Proved by the definitions of executions, traces and $\levup$ operator.
\end{proof}

\begin{theorem}
\label{th:sub_imp_in2}
Let $\AutB_1$ and $\AutB_2$ be comparable WAs with $\AutB_1 \leq \AutB_2$.
Let $\AutA$ be a WA compatible with each of $\AutB_1$ and $\AutB_2$. Then $\AutA [\AutB_1]$ and $\AutA [\AutB_2]$ are comparable and $\AutA [\AutB_1] \leq \AutA [\AutB_2]$.
\end{theorem}
\begin{proof}
Let $\alpha$ be an execution of $\AutA[\AutB_1]$. By lemma \ref{lm:inexec}, two executions $\alpha_A, \alpha_1$ exist, such that $\alpha_A\in\hexec{\AutA}$, $\alpha_1\in\hexec{\AutB_1\levup}$ and:
$\alpha\restr(A_A,V_A\setminus C_{A1})=\alpha_A\restr(A_A,V_A\setminus C_{A1})$,
$\alpha\restr(A_1,V_1\setminus C_{A1})=\alpha_1\restr(A_1,V_1\setminus C_{A1})$, 
$\alpha\restr(\emptyset, C_{A1})=\alpha_A\restr(\emptyset, C_{A1})-\alpha_1\restr(\emptyset, C_{A1})$,
with $C_{A1}=U_{wA}\cap Y_{w1}$.
By lemma \ref{lm:padbuild} we can take paddings of $\alpha, \alpha_A, \alpha_1$ such that the $j^{th}$ trajectory has the same length for all $j$. Let these paddings be $\padex, \padex_A, \padex_1$ respectively with $\padex=\tau_0 a_1 \tau_1 a_2 \tau_2 a_3 \ldots$, $\padex_A=\tau_{0A} a'_{1} \tau_{1A} a'_{2} \tau_{2A} a'_{3} \ldots$ and $\padex_1=\tau_{01} a''_{1} \tau_{11} a''_{2} \tau_{21} a''_{3} \ldots$.
Since $\AutB_1 \leq \AutB_2$ (hence $(\AutB_1\levup) \leq (\AutB_2\levup)$ by lemma \ref{lm:levupimp}) and by compatibility, we can find an execution $\alpha_2$ of $\AutB_2$ with the same trace of $\alpha_1$ and a padding of $\alpha_2$ following lemma \ref{lm:padbuild}. We write $\padex_2=\tau_{02} a'''_{1} \tau_{12} a'''_{2} \tau_{22} a'''_{3} \ldots$. 
By the definition of composition the execution of $\AutA[\AutB_2]$ obtained by $\padex_A$ and $\padex_2$ will be $\padex'=\tau'_0 b_1 \tau'_1 b_2 \tau'_2 b_3 \ldots$, where
$\tau'_j\restrrange (V_A\setminus C_{A2})=\tau_{jA}\restrrange (V_A\setminus  C_{A2})$,
$\tau'_j\restrrange (V_2\setminus C_{A2})=\tau_{j2}\restrrange (V_2\setminus  C_{A2})$,
$\tau'_j\restrrange  C_{A2}=\tau_{jA}\restrrange  C_{A2}-\tau_{j2}\restrrange C_{A2}\}$,
where $C_{A2}=U_{wA}\cap Y_{w2}$. This is valid even if the trajectories in the padded executions have not the length of the original trajectories, by definition of prefix of a trajectory and prefix closure of trajectories in a WA.
Actions $b_i$ might be different, but by construction, compatibility and lemma \ref{lm:padtrace} we have that $\padex'$ has the same trace of $\padex$ hence of $\alpha$. Indeed the (padded) executions can differ only in their internal variables (state), but they do not influences the traces (external variables). For this reason we can state that $\htraces{\AutA[\AutB_1]}\subseteq\htraces{\AutA[\AutB_2]}$, hence, by definition of implementation, $\AutA [\AutB_1] \leq \AutA[\AutB_2]$.
\end{proof}

{\bf Notation for proofs}: The interested reader can find all the definitions and notation in \cite{gandalf12,tec_rep}.
\begin{definition}
A {\em padded execution} of a WA $\AutA$ is an $(A\cup\{\padact\},V)-$sequence $\padex=\tau_0 a_1 \tau_1 a_2 \tau_2 a_3\ldots$ such that if $a_i=\padact$ then $\lstate{\tau_{i-1}}=\fstate{\tau_i}$.
\end{definition}
\begin{definition}{\em Padding.}\\
We call {\em padding} of an execution $\alpha$ any padded execution obtained by $\alpha$ by extending the actions set with $\padact$.
\end{definition}
%
%
%
%
%
\begin{lemma}\label{lm:padtrace}
Let $\alpha$ be an execution of $\AutA$ and $\padex$ a padding of $\alpha$, then $\htrace{\alpha}=\htrace{\padex}$.
\end{lemma}
%
%
\begin{lemma}\label{lm:padbuild}
Given $n$ executions, it is always possible to find $n$ paddings of these executions such that all corresponding trajectories have the same length.
\end{lemma}

%
%

\section{Case Study}

In this section we use the above introduced theory to describe a scenario based on the work in \cite{polycarpou}. 
We consider an environment, later called field. On this field there are $M$ targets, of which $M_k\subset M$ are known. Over the field $n$ UAVs (Unmanned Aerial Vehicles) fly, which are supposed to detect all the targets and engage them. We are not interested in the cooperative search policy, which is described in  \cite{polycarpou}. We want to prove, using this example, that our formal model is reliable in describing the characteristics of such scenario, preserving some results of composability. 
The WAs reported here are not the only ones that can be used to represent the same scenario of course, we used them because they are the most straightforward to design.
To represent WAs we use a variant of the TIOA language (see \cite{TIOA}), with some extensions for hybrid systems (\cite{helicopter}) and some new tools for world variables. Note that world variables are always described using their trajectories in time and space, i.e. they are described for any instant of time $t$ and any point in space $p$. We assume an underlying metric space $\Reals^2$, where points are indicated by $p\in\Reals^2$. Obviously $p$ is given by a couple of coordinates $(x,y)$. We do not use $\Reals^3$ because we are interested in the position of the targets on the field and in the projection of the position of the UAVs on the field, not in their altitude. We also assume that the area occupied by the field is the entire space $\Reals^2$.

Since the description of the field characteristics is very dependent from the nature of the targets inside it, we start by describing targets represented in fig. \ref{target}. 
Usually what determines the interaction with the ground is the class of the target. For example if we are in an environment surveillance scenario, some classes might be: fire, flood, steam, etc. A very general parameter might be the color: if a UAV is flying over a certain field, which we suppose to be green, and we see a black spot, then we conclude that something wrong is happening, hence a black target might be there. 


The local internal variables of a Target are: position $p_T$, orientation $\phi$, a Fail signal used by the target to {\em delete} itself when engaged. The world input is compression $k$, given by the ground when the target is engaged. For example, if the target is a fire and the UAV is pouring some water, the field will react by getting colder and change the pressure in the meanwhile, until the fire is put off. We will generally call compression the reaction of the ground to the engage signal. The target world output is its color $\xi$. Each target has an ID given by its parameter IDT and a class given by its color $T_C$, which is still a parameter. 
A function $f$ is defined for targets, giving the surface of the field occupied by the target area. We do not specify this function because it is irrelevant to the scope of this paper, but it is calculated starting from the target position $p_T$ and its orientation angle $\phi$. 
The target has an internal action {\em delete} occurring when the field compression reaches a certain value $k_{max}$. Its effect is to put the Fail signal on and then change the color of the target location to the color of the field, which we assume to be green. 
Note that the Target has two states: in one it is alive and its Fail signal is false, in the other one it is deleted and its Fail signal is true. The switching between the states occurs when the action {\em delete} arises.


The Field WA is represented in fig. \ref{field}. 
The world internal variables of Field are color $C$ and compression $K$. The world input variables are: color of the target $\xi$ and the engage signal $e$ from UAVs. The world output variables refer compression $k$ and color $c$. The Field has only one state in which its color is a replica of the color variable given by the targets, its compression increases when the engage signal is activated and its outputs are used to communicate with targets (compression) and UAVs (color).

UAVs have a more complex structure represented in fig. \ref{uav}. Each UAV is specified by an ID (IDU) and a color ($T_U$) associated with the kind of target it can engage. Note that not all UAVs can engage all targets, but only the ones with the same color. Again we are not going into the detail of levels greater than 0. 
UAVs local internal variables are: position $p_U$, speed $v$, heading angle $\psi$, performing task $\ms{tk}$, actual assignment assign, position given by the sensor footprint $s_{fp}$. The UAV can perform one of the following tasks: search ($S$), confirm ($C$), engage ($E$). The search task is the default one: the UAV moves following a constant trajectory and always searching for targets. When it sees with a certain probability a target (color different from green) it updates a map of targets, and its task goes to confirm. If the target is of the same color of the one given by $T_U$, then the UAV can compete for engage it. If it wins the competition its task moves to engage. Similarly for the assignments: a free assignment means that the UAV has not been associated to any target, i.e. it has not been elected to engage a target, but it is not even competing to engage a target; a competing assignment means that a UAV is competing with other UAVs to engage a certain target; a committed assignment means that the UAV has been elected to engage the target. The sensor footprint gives the position in which the UAV is looking for a target.
The local input variables are: probability $Pin_t$ of presence of a target broadcasted by the UAVs, probability $Pin_\phi$ of orientation of a target broadcasted by the UAVs, cost$_{in}$ array of costs for reaching a target given by the other competing UAVs.
The local output variables are: probability $P_t$ of presence of a target, probability $P_\phi$ of orientation of a target, cost of reaching the target for which the UAV is competing; map of target presence probability for each position $m_p$, map of task for each position $m_t$, map of target kind for each position $m_k$, map of target orientation estimate for each position $m_\phi$, map of target assignment $m_a$. All the maps are broadcasted and updates using a function {\em update} which puts a value on a specific position of the map. We are not going into the detail neither of the map model, nor of the update function.
The UAV world input is the color of the field $c$, its world output is the engage signal $e$.

\begin{figure}[htbp]
\begin{scriptsize}

{\bf type} Rad = $\Reals| 2\pi$


{\bf type} Color = \{green, $\chi_1$, $\chi_2$, $\chi_3$\}
\\

{\bf worldautomaton} Target (IDT: Real,$T_C$: Color)\\

{\bf world variables LEVEL 0}

\qquad	{\bf input}  $k$: Real 

\qquad	{\bf output} $\xi$: Color

{\bf local variables LEVEL 0}

\qquad	{\bf internal} $\phi$: Rad, $p_T$: Real$^2$, 

\qquad\qquad\qquad   Fail: Bool $:=$ false

{\bf actions}

\qquad      {\bf internal} {\em delete}

{\bf transitions}

\qquad {\bf internal} {\em delete}

\qquad {\bf pre} $\exists p\in f(\phi,p_T)$ s.t. $k(t,p) \geq k_{max}$

\qquad {\bf eff}  Fail $=$ true

{\bf trajectories}

\qquad   $\xi(t,p)=\left\{\begin{array}{ll}
T_C & \mbox{if } p\in f(\phi,p_T) \land \neg\mbox{Fail}\\
\mbox{green} & \mbox{otherwise}
\end{array}\right.
$ 


\caption{Target world automaton.}
\label{target}
\end{scriptsize}
\end{figure}

\begin{figure}[htbp]
\begin{scriptsize}

{\bf worldautomaton} Field \\

{\bf world variables LEVEL 1}

\qquad	{\bf internal} $C$: Color, $K$: Real $:= 0$ 

\qquad	{\bf input} $\xi$: Real, $e$: Bool

\qquad	{\bf output} $k$:Real, $c$: Color 





{\bf trajectories}

\qquad  $C(t,p)=\xi(t,p)$;


\qquad   $\dot{K}(t,p)=e(t,p)?1:0$;


\qquad   $c(t,p)=C(t,p)$; 

\qquad   $k(t,p)=K(t,p)$.


\caption{Field world automaton.}
\label{field}
\end{scriptsize}
\end{figure}

\begin{figure}[htbp]
\begin{scriptsize}



{\bf type} Task = $\{S, C, E\}$ 

{\bf type} Assignment = $\{$free, competing, committed$\}$ 
\\

{\bf worldautomaton} UAV (IDU: Real, $T_U$: Color)\\

{\bf world variables LEVEL 0}

\qquad	{\bf input} $c$: Color, 

\qquad	{\bf output} $e$: Bool $:=$ false

{\bf local variables LEVEL 0}

\qquad	{\bf internal} $p_U$: Real$^2$, $v$: Real$^2$, $\psi$: Rad, 


\qquad\qquad\qquad      $\ms{tk}$: Task $:= S$,  

\qquad\qquad\qquad      assign: Assignment $:=$ free, 

\qquad\qquad\qquad      $s_{fp}$: Real$^2$ 

\qquad  {\bf input}  $Pin_t$: $\{0,1\}$, $Pin_\phi$: $\{0,1\}$, cost$_{in}$: array of Real

\qquad  {\bf output}  $P_t$: $\{0,1\}$, $P_\phi$: $\{0,1\}$, 
cost: Real, 

\qquad\qquad\qquad      $m_p$: Map, 
$m_t$: Map, 
$m_k$: Map, 

\qquad\qquad\qquad      $m_\phi$: Map, 
$m_a$: Map 

{\bf actions}

\qquad      {\bf internal} {\em search}, {\em confirm}, {\em engage}, {\em compete}, {\em commit}

{\bf transitions}

\qquad   {\bf internal} {\em search}

\qquad   {\bf pre} $P_t(s_{fp},t)\leq p_s \lor \exists i \mbox{ s.t. cost}(t)\geq \mbox{cost}_{in}^i(t)$

\qquad\qquad  $\lor\ c(t,s_{fp})=$ green   

\qquad   {\bf eff}   $\ms{tk}=S$; assign $=$ free;


\qquad   {\bf internal} {\em confirm}

\qquad   {\bf pre} $P_t(s_{fp},t)> p_s$   
$\land\ c(t,s_{fp})\neq$ green 

\qquad   {\bf eff}    $\ms{tk}=C$; 


\qquad   {\bf internal} {\em compete}

\qquad  {\bf pre} $c(t,s_{fp})=T_U$

\qquad  {\bf eff}  assign $=$ competing;  

\qquad   {\bf internal} {\em commit}

\qquad  {\bf pre}  cost$(t)=\min_i$ cost$_{in}^i(t)$ 

\qquad  {\bf eff}  assign $=$ committed; 

\qquad   {\bf internal} {\em engage}

\qquad   {\bf pre} $P_t(s_{fp},t)> p_e$   

\qquad   {\bf eff}   $\ms{tk}=E$;  

{\bf trajectories}

\qquad  $\dot{p_U}(t)=v(t)$;

\qquad  $\dot{\psi}\leq\eta$; 

\qquad  $\dot{v}=0$;

\qquad  $e(t,p)= (\ms{tk}=E)?$ true: false;

\qquad  $P_\phi(s_{fp},t)=g(P_\phi(s_{fp},t^-),Pin_\phi(s_{fp},t),c(t,s_{fp}),\psi(t))$;

\qquad  $m_\phi(s_{fp},t)=$ {\em update}$(P_\phi(s_{fp},t))$; 

\qquad  $P_t(s_{fp},t)=h(P_t(s_{fp},t^-),Pin_t(s_{fp},t),c(t,s_{fp}))$;

\qquad  $m_p(s_{fp},t)=$ {\em update}$(P_t(s_{fp},t))$; 

\qquad  cost$(t) = r(|P_U-s_{fp}|)$;

\qquad $m_k(s_{fp},t)=$ {\em update}($c(t,s_{fp})$);

\qquad $m_t(s_{fp},t)=$ {\em update}(task); 

\qquad  $m_a(s_{fp},t)= ${\em update}(assign);


\caption{UAV world automaton.}
\label{uav}
\end{scriptsize}
\end{figure}


We now want to put a Target inside Field. We verify that the inplace compatibility is valid. First of all we lift the level of Target (and hence of all its variables and actions) to 1. We did not define variables at levels greater than 1, hence the first condition for compatibility is valid. The second condition is on world outputs: $Y_F=\{k,c\}$, $Y_T=\{\xi\}$ implies that $Y_T\cap Y_F=\emptyset$. Since Field has no actions, the third condition is also verified. Field and Target have no output variables and no output actions, so the last two conditions for compatibility are also verified.
The inplacement of Target in Field is the object described in fig. \ref{fieldtarget}. The reader can notice that condition 9 of Def. \ref{def:inplace} is respected, since world variable $\xi\in U_{w1}\cap Y_{w3}$.


\begin{figure}[htbp]
\begin{scriptsize}

{\bf worldautomaton} Field[Target(IDT,$T_C$)]\\

{\bf world variables LEVEL 1}

\qquad	{\bf internal} $C$, $K:=0$

\qquad	{\bf input} $\xi$, $e$

\qquad	{\bf output} $k$, $c$

{\bf local variables LEVEL 1}


\qquad     {\bf internal}        $\phi$, $p_T$, Fail $:=$ false

{\bf actions}

\qquad      {\bf internal} {\em delete}

{\bf transitions}

\qquad {\bf internal} {\em delete}

\qquad {\bf pre} $\exists p\in f(\phi, p_T)$ s.t. $k(t,p) \geq k_{max}$

\qquad {\bf eff}  Fail $=$ true

{\bf trajectories}


\qquad  $C(t,p)=\neg\mbox{Fail}?\xi(t,p) : \mbox{green}$

\qquad   $\dot{K}(t,p)=e(t,p)?1:0$;


\qquad   $c(t,p)=C(t,p)$; 

\qquad   $k(t,p)=K(t,p)$.


\caption{Field[Target]}
\label{fieldtarget}
\end{scriptsize}
\end{figure}

\begin{figure}[htbp]
\begin{scriptsize}

{\bf worldautomaton} FalseField\\

{\bf world variables LEVEL 1}

\qquad	{\bf internal}  $\ms{fC}$: Color, $\ms{fK}$: Real $:=0$

\qquad	{\bf input} $e$: Bool

\qquad	{\bf output} $c$: Color, $k$: Real

{\bf local variables LEVEL 1}


\qquad {\bf internal} $\theta$: Rad, $\ms{fp}$: Real, Del: Bool $:=$ false, $\chi$: Color

{\bf actions}

\qquad      {\bf internal} {\em change}

{\bf transitions}

\qquad {\bf internal} {\em change}

\qquad {\bf pre} $\exists p\in f(\theta, \ms{fp})$ s.t. $k(t,p) \geq k_{max}$

\qquad {\bf eff}  Del $=$ true

{\bf trajectories}

\qquad   $\ms{fC}(t,p)=(p\in f(\theta, \ms{fp}) \land \neg\mbox{Del})? \chi : \mbox{green}$;

\qquad   $\dot{\ms{fK}}(t,p)=e(t,p)?1:0$;

\qquad   $c(t,p)=\ms{fC}(t,p)$; 

\qquad   $k(t,p)=\ms{fK}(t,p)$.


\caption{FalseField}
\label{falsefield}
\end{scriptsize}
\end{figure}

We now define system FalseField as in fig. \ref{falsefield}. 
We want to prove that FalseField is equivalent to Field[Target(IDT,$T_C$)], i.e. a bisimulation exists between the two automata. Consider a state of Field[Target(IDT,$T_C$)] in which Fail=false ($k(t,p)<k_{max}, \forall p\in \Reals^2$), hence $C(t,p)=\xi(t,p)$, i.e. $C(t,p)=T_C, p\in f(\phi,p_T)$ and green everywhere else. Also in this state $e(t,p)=$ false, $\forall p\in \Reals^2$, hence $\dot{K}(t,p)=0$. The bisimilarity relation is the identity for all internal variables (both world and local). The corresponding state of FalseField is given by Del=false, hence $fC(t,p)=\chi, \forall p\in f(\theta,\ms{fp})$, where $\theta$ has the same value of $\phi$, $\ms{fp}$ has the same value of $p_T$ and $\chi$ has the same value of $T_C$. Since $e(t,p)=$ false, $\forall p\in \Reals^2$, then $\dot{\ms{fK}}(t,p)=0$. Hence if the color of Target is $\chi_2$, i.e. $T_C=\chi=\chi_2$, the output variables of both WAs will be $c=\chi_2$ and $k=$ CONST. 
Now if $e(t,p)=1$ for some $p$, then $\dot{K}(t,p)=1$, i.e. $K$ starts increasing, and so does $k$. After a certain time $t$ it reaches the value of $k_{max}$ and the action {\em delete} occurs. Hence in the new state Fail=true, i.e. $C(t,p)=$ green for $p\in f(\phi,p_T)$. Similarly, for FalseField, $e(t,p)=1$ implies that $\dot{\ms{fK}}(t,p)=1$, i.e. $\ms{fK}$ starts increasing, and so does $k$. When $k$ reaches $k_{max}$, action {\em change} occurs, and Del=true. In this new status $\ms{fC}(t,p)=$ green for $p\in f(\theta,\ms{fp})$. In both cases $c=$ green.

What we want to prove now is that (Field[Target])[UAV] is equivalent to FalseField[UAV]. For lack of space these two automata are not represented here, but the reader can easily follow the inplacement procedure to design them. 
We start from the following state in Field[Target][UAV] at time $t^*$:
 Fail=false,
$c(t*,p*)=T_U=\chi_2$ for a certain position $p*\in\Reals^2$,
$s_{fp}=p*$,
$\ms{tk}=C$,
assign=committed.
The corresponding state of FalseField[UAV] is:
Del=false,
$c(t*,p*)=T_U=\chi_2$,
$s_{fp}=p*$,
$\ms{tk}=C$,
assign=committed.
Hence equivalence by identity is preserved. 
At next instant of time, $t*^+$, $P_t(p*,t*^+)> p_e$, and action {\em engage} occurs with the effect of letting $\ms{tk}=E$. As a consequence $e(t*^+,p*)=$ true. In (Field[Target])[UAV], this implies that $\dot{K}(t*^+,p*)=1$. Output variable $k$ will start increasing, following internal variable $K$ until it reaches value $k_{max}$ and activating action {\em delete} with the consequence of variable Fail becoming true and world output $c$ becoming green. Analogously in FalseField[UAV] $\dot{\ms{fK}}(t*^+,p*)=1$ and output $k$ starts increasing until it reaches value $k_{max}$ such activating action {\em change} with the consequence of variable Del becoming true and world output $c$ becoming green.
The above execution represents the case in which a UAV $i$ is associated with a certain target $j$ having the same color, i.e. $T_C(j)=T_U(i)$. The UAV sees where the color is and activates the engage signal to delete target $j$, that after a certain time fails and sets the color in its location to green.

\section{Concluding Remarks}\label{sec:conclusion}

In this paper we proposed an extension of the model introduced in \cite{gandalf12} to provide an explicit and natural representation of the fact that objects move in a world that they can observe and modify. 
Besides the classical signals that automata send to each other via discrete communication events or shared continuous variables, the automata we introduced, called World Automata (WAs), can communicate implicitly by affecting their surrounding world and observing the effects on the world of the activity of other automata. This mechanism for interaction turns out to be adequate for compositional analysis, which is one of the main features of HIOAs that we wanted to keep in the extended model. Moreover we represented hierarchical nested worlds, by introducing a notion of levels of variables. Indeed we extended the notion of parallel composition of \cite{gandalf12} to WAs, keeping compositionality results. We introduced another operator, called inplacement, to represent the hierarchical composition of WAs.
We presented in this paper an example to show the effectiveness of the theory, but another reality-based application can be found in \cite{adhs12}. The simulation tools are under study.
Future research directions include the ability to describe scenarios where automata are created and destroyed and where communication links change dynamically. One approach for dynamic communication that we find promising is to associate each state and output action with an open set of the underlying space $\M$ and interact only with automata that lie in such open set, which we can call neighborhood. 
The natural extension to this problem is the ability of WAs to achieve selective communication, by choosing if they want to broadcast information to every other agent in the neighborhoods or only to some of them.

%
\bibliographystyle{eptcs}
\bibliography{wa}  

\end{document}